\documentclass[journal]{IEEEtran}

\usepackage[latin1]{inputenc}
\usepackage{graphicx}
\usepackage{amssymb,amsmath,amsfonts,amsthm}
\usepackage{algorithm,algorithmic}
\usepackage{cite}
\usepackage{float}
\usepackage{graphics}
\usepackage{subfigure}
\usepackage{xspace}
\usepackage[usenames,dvipsnames]{color}
\usepackage{epsfig, hyperref}
\usepackage{amsmath}
\newtheorem{theorem}{Theorem}
\newtheorem{statement}{Statement}
\usepackage{amsmath}

\begin{document}

\title{Retrodirective Large Antenna Energy Beamforming in Backscatter Multi-User Networks}

\author{Ioannis Krikidis,~\IEEEmembership{Senior Member,~IEEE}
\thanks{I. Krikidis is with the Department of Electrical and Computer Engineering, University of Cyprus, Nicosia 1678 (E-mail: {\sf krikidis@ucy.ac.cy}).}}

\maketitle

\begin{abstract}
In this letter, we study a new technique for energy beamforming (EB) in multi-user networks, which combines large antenna retrodirectivity at the transmitter side with signal backscattering at the energy receivers. The proposed technique has low complexity and achieves EB without any active operation at the receivers or complicated signal processing techniques at the transmitter. Since the average harvested energy depends on the backscattering coefficients, we investigate different reflection policies for various design objectives. The proposed policies are analyzed from a system level standpoint by taking into account spatial randomness.      
\end{abstract}
\vspace{-0.2cm}

\vspace{-0.5cm}
\section{Introduction}

\IEEEPARstart{W}{ireless} power transfer (WPT) through dedicated radio-frequency (RF) radiation has been introduced as a promising technology for energizing low-power devices \cite{CLE}. It is relevant for future mobile networks, where  the majority of billion of devices will be low-powered and autonomous.  To increase the range of WPT and boost the end-to-end power transfer efficiency, several techniques have been proposed in the literature that incorporate new circuit designs,  network architectures and recent advances in signal processing such as energy beamforming (EB) \cite{RUI0}. 

EB refers to systems with multiple antennas and enables focusing the transmitted energy to the direction of the receivers. The implementation of EB requires channel state information (CSI) at the energy transmitter (ET), which becomes a critical assumption for WPT systems, since the energy capabilities of the energy receivers (ERs) are limited. To tackle this problem, most of works consider a quantized feedback and/or exploit the channel reciprocity in association with different channel learning methods  \cite{RUI2, RUI3}.  In \cite{RUI1}, the authors introduce an alternative EB scheme that does not require channel estimation, by exploiting  the concept of large antenna retrodirectivity. However, its implementation requires signal transmission from the ERs, which could be costly for WPT scenarios. On the other hand, the work in \cite{YAN} studies EB in a backscatter communication system, where CSI is acquired by leveraging the backscatter signal that is reflected back from the ERs. A beneficial combination of backscatter communication and retrodirective beamforming for EB is introduced in \cite{HUA}; however, a theoretical analysis is not provided.

In this letter, we study a new EB technique for multi-user networks that combines large antenna retrodirectivity with signal backscattering. The proposed technique overcomes the limitations of the previous works in \cite{RUI1}, \cite{YAN}, and achieves EB without CSI estimation at the ET and without signal transmission at the ERs. We show that the achieved energy harvesting performance depends on the backscattering reflection coefficients. Different reflection policies that correspond to different complexities and energy harvesting performances are investigated. The proposed EB technique and the associated reflection policies are analyzed from a system level point of view by using stochastic geometry tools.

\vspace{-0.35cm}
\section{System model}\label{system_model}

We consider a single-cell consisting of an ET and multiple randomly deployed ERs, where the coverage area is modeled as a disc $\mathcal{D}$ with radius $\rho$ and an exclusion zone of radius $\xi>1$. The exclusion zone captures potential safety issues i.e., ERs are prohibited to be close to the ET \cite{XIA}. The ET employs a large number of antennas, denoted by $M$, while the ERs are equipped with single antennas.  We assume that the ERs are located according to a homogeneous Poisson point process (PPP) $\Phi=\{x_i\}$ with a density $\lambda$ that ensures $M\gg K$, where $K$ is the number of ERs in $\mathcal{D}$ (i.e., Poisson distribution with mean $\lambda \pi (\rho^2 -\xi^2)$ \cite{HAN2}).  

The ET is connected to the power grid and transmits with a fixed power $P_t$; its antenna array has (digital) retrodirectivity properties and thus it can re-transmit a signal back along the spatial direction of the incoming signal without any a-priori knowledge of its point of origin \cite{MAS}. On the other hand, the ERs do not have any local power supply and they power their functionalities through WPT by using the signal of the ET. Specifically, the ERs operate in a backscattering mode, where part of the received signal is reflected back by adapting the level of antenna impedance mismatch; let $0\leq \beta_i \leq 1$ be the reflection coefficient for the $i$-th ER. 

The large-scale path loss is assumed to be $d_i^{-\alpha}$, where $d_i\geq \xi$ denotes the Euclidean distance between the $i$-th ER and the ET, and $\alpha>2$ is the path-loss exponent; the exclusion zone ensures a bounded path loss region with radius $\xi$. We assume a Rayleigh block fading \cite{RUI1,XIA}, where the channel between the ET and the $i$-th ER is an uncorrelated Gaussian vector with elements having zero mean and unit variance i.e., $\pmb{f}_i=[f_{i,1},\ldots, f_{i,M} ]^T\sim \mathcal{CN}(\pmb{0},\pmb{I}_M)$.  We assume channel reciprocity and thus the channels for uplink and downlink are equivalent and constant during one block, but vary independently from block to block. It is worth noting that  the proposed scheme becomes less interesting  
for scenarios with strong line-of-sight components (e.g., for Rice fading channel, the omnidirectional EB approaches the performance of the perfect EB  as the Rice factor increases).

\vspace{-0.3cm}
\section{Retrodirective EB with backscattering}

The proposed EB technique combines a retrodirective beamforming at the ET with a passive backscattering process at the ERs. The EB technique consists of two phases, as described below.

\vspace{-0.3cm}
\subsection{Backscattering phase} 

The first phase enforces the ERs to notify their spatial direction to the ET by using their backscattering capabilities. More specifically, the ET generates and broadcasts an unmodulated single-tone waveform $x(t)=\sqrt{2P_t} \cos(2\pi f_c t)$ with $0<t<\tau$, where $f_c$ is the carrier frequency and $\tau$ is the duration of the waveform\footnote{The proposed scheme requires the orthogonality between single-tone waveform transmission and retrodirectivity at the ET. An appropriate design of the parameters $\tau$, $\xi$ as well as a time delay between reception/reflection at the ERs ensures that the duration of the waveform is shorter than the round-trip delays.}; the power is symmetrically divided to $M$ antennas.  Each ER operates in the backscattering mode and reflects back a part of the received signal without further processing. The received signal (complex baseband representation) at the ET is given by
\vspace{-0.2cm}
\begin{align}
\pmb{y}(t)&=\sum_{x_k \in \Phi}\sqrt{\beta_k (P_t /M)d_k^{-2\alpha}}g_k \pmb{f}_k+\pmb{n}(t),
\end{align}
where $g_k=\sum_{m=1}^M f_{k,m}$ denotes the equivalent channel for the downlink (sum of $M$ complex Gaussian random variables with unit variance), $\pmb{f}_k$ is the uplink channel due to the reciprocity, and $\pmb{n}(t)$ is the additive white Gaussian noise (AWGN) vector with $\pmb{n}(t)\sim \mathcal{CN}(\pmb{0},\sigma^2 \pmb{I}_M)$. Note that the path-loss exponent is doubled as a result of the dyadic backscatter channel.

\vspace{-0.3cm}
\subsection{EB phase} 

In the second phase, the ET exploits large antenna retrodirectivity and uplink/downlink channel reciprocity and effectively forms coherent beams back to the ERs. More specifically, the ET performs matching filtering on $\pmb{y}(t)$ to detect the phase of the received signal; then it transmits an energy signal with a power $P_t$ by employing phase conjugation. The energy transmit signal (baseband) can be written as $\pmb{x}_e=(\pmb{w}+\tilde{\pmb{n}})^H/||\pmb{w}+\tilde{\pmb{n}}||$, where $\pmb{w}=\sum_{x_k \in \Phi}\sqrt{\beta_k (P_t /M)d_k^{-2\alpha}}g_k^* \pmb{f}_k^H$ and $\tilde{\pmb{n}}\sim \mathcal{CN}(\pmb{0},(\sigma^2/\tau) \pmb{I}_M)$ denotes the noise at the output of the matched filter. Due to the channel reciprocity, the received signal at the $i$-th ER can be expressed as $y_i= \sqrt{P_t d_i^{-\alpha}}\pmb{x}_e \pmb{f}_i+z$, where $z$ is the AWGN component.  By using the energy conversion law and by ignoring energy harvesting from AWGN, the harvested energy at the $i$-th ER becomes 
\begin{align}
&Q(\pmb{\beta},d_i)= \zeta |y_i|^2 =\underbrace{\zeta P_t d_i^{-\alpha}}_{Q_{\text{OM}}(d_i)}+\underbrace{\frac{\zeta P_t M d_i^{-3\alpha}\beta_i|g_i|^2}{\sum_{x_k \in \Phi} d_k^{-2\alpha} \beta_k|g_k|^2+\frac{M \sigma^2}{P_t \tau}}}_{Q_{\text{RE}}^{\Phi}(\pmb{\beta}, d_i)}, 
\label{harvesting}
\end{align}
where $\zeta$ denotes the RF-to-DC energy conversion efficiency\footnote{A linear power channel is sufficient to demonstrate the proposed EB technique \cite{RUI2,RUI3}; wireless power nonlinearities are beyond the scope of this letter.}. 
\begin{proof}
The proof follows similar arguments with \cite[Lemma 3.1]{RUI1} and is based on the asymptotic massive multiple-input multiple-output (MIMO) expressions \cite{LIM} i.e., $\frac{1}{M}||\pmb{f}_i||^4\rightarrow M+1$, $\frac{1}{M}||\pmb{f}_i||^2\rightarrow 1$, $\frac{1}{M}\pmb{f}_k^H{\pmb{f}_i}\rightarrow 0$ (with $k\neq i$), $\frac{1}{M}\pmb{f}_k^H{\tilde{\pmb{n}}}\rightarrow 0$, $\frac{1}{M}||\pmb{f}_k^H\tilde{\pmb{n}}||^2\rightarrow \frac{\sigma^2}{\tau}$, and $\frac{1}{M}||\tilde{\pmb{n}}||^2\rightarrow \frac{\sigma^2}{\tau}$.
\end{proof}
The harvested energy consists of two main components i) an omnidirectional component $Q_{\text{OM}}$ which only depends on the location of the ER, and b) a retrodirective component $Q_{\text{RE}}^{\Phi}$ that depends on the reflection coefficients $\beta_i$ of the ERs.

\vspace{-0.4cm}
\section{Average harvested energy}
In this section, we study the impact of reflection coefficients on the achieved performance and we investigate different reflection strategies.  

\vspace{-0.4cm}
\subsection{Distance-inversion backscattering (DIB)}

In this strategy, we enforce the average harvested energy associated with the retrodirective component to be equal for each ER. To satisfy this fairness requirement, we assume that the reflection coefficient is a function of the location of each ER i.e., $\beta_i=(d_i/\rho)^{2\alpha}$; this assumption ensures that the reflection coefficient increases with the distance and takes the maximum value (full reflection) at the edge of the cell. In this case, the expression in \eqref{harvesting} can be simplified as
\begin{align}
&Q(\pmb{p},d_i)\approx \zeta P_t d_i^{-\alpha}\!\! \left [\!1\!+\!\frac{M|g_i |^2}{|g_i|^2+\sum_{x_k \in \Phi, i\neq k} |g_k|^2} \right]\;\text{for}\;\sigma^2\!\rightarrow 0,
\end{align}
where $\pmb{p}=[(d_1/\rho)^{2\alpha},\ldots,(d_K/\rho)^{2\alpha}]$ and the average harvested energy is given by 
\begin{subequations}
\begin{align}
&\mathcal{Q}_{\text{DIB}}=\zeta P_t \mathbb{E}_{\Phi}[d_i^{-\alpha}]\left( 1+M\mathbb{E} \left[ \frac{|g_i |^2}{|g_i|^2+\sum_{x_k \in \Phi, i\neq k} |g_k|^2}   \right]\right) \nonumber \\
&=\zeta P_t \int_{0}^{2\pi}\int_{\xi}^{\rho}x^{1-a} f_d(x) d\theta dx \bigg(1+M\mathbb{E}[\mathcal{Z}]\bigg)   \label{q1} \\
&=   \underbrace{\frac{2\zeta P_t (\rho^{2-\alpha}-\xi^{2-\alpha}) }{(2-\alpha)(\rho^{2}-\xi^2)}}_{\Lambda(\xi,\rho)} \left(1+\frac{M}{\lambda \pi (\rho^2-\xi^2)}\right), \label{q2} 
\end{align}
\end{subequations}
where $f_d(x)=1/\pi(\rho^2-\xi^2)$ denotes the probability density function (PDF) of each point in the disc $\mathcal{D}$, and \eqref{q2} is based on the expected value of Beta distribution with shape parameters $M$ and $(\mathbb{E}[K]-1)M$, respectively, i.e., $\mathcal{Z}\sim Beta(M,[\lambda \pi (\rho^2-\xi^2)-1] M)$.

\vspace{-0.3cm}
\subsection{Full backscattering (FB)}
The FB reflection strategy refers to systems without further intelligence, where the ERs fully reflect back their received signals i.e., $\beta_i=1$ ($\forall\; x_i\in \Phi$). This scheme does not require any coordination between the ERs and therefore provides implementation simplicity. For the average harvested energy, we state the following theorem. 

\begin{theorem}
The average harvested energy achieved by the FB scheme, is given by
\begin{align}
\mathcal{Q}_{\text{FB}}&= \Lambda(\xi,\rho)+Q_{\text{FB}}(\xi,\lambda).
\end{align}
\end{theorem} 
\begin{proof}
The first term is associated with the omnidirectional component of the harvested energy and is similar to the first term of the DIB scheme (see \eqref{q2}), $Q_{\text{FB}}(\xi,\lambda)$ refers to the retrodirective component and is given in an integral form in the Appendix. 
\end{proof}
For the special cases of highly dense/sparse networks with $\sigma^2\rightarrow 0$, we introduce the following statement. 
\begin{statement}
For $\sigma^2\rightarrow 0$, the average harvested energy for highly dense ($\lambda\rightarrow \infty$) and sparse networks ($\lambda\rightarrow 0$) converges to the one achieved by the non-backscattering ($\beta_i=0$) and beamforming (perfect CSI) schemes, respectively, i.e., 
\begin{align}
&\mathcal{Q}_{\text{FB}}\rightarrow \Lambda(\xi,\rho)\;\;\;\;\;\;\text{for}\;\; \lambda\rightarrow \infty,\;\; \text{[Non-backscattering]},
\\
&\mathcal{Q}_{\text{FB}}\rightarrow M\Lambda(\xi,\rho)\;\;\text{for}\;\; \lambda\rightarrow 0,\;\; \text{[Beamforming]}. 
\end{align} 
\end{statement}

\vspace{-0.8cm}
\subsection{Binary backscattering schemes}

In this strategy, the ERs operate either in a full-backscattering mode or in a non-backscattering mode i.e., $\beta_i\in \{0,1\}$. Through this binary reflection mechanism, we can control the harvested energy related to the retrodirective component and provide design flexibility. The selection of the operation mode is driven by two criteria, which are given in the following discussion.       

\subsubsection{Distance-based binary backscattering (DBB)}
The first criterion takes into account the spatial location of the ERs; ERs that are located at a distance higher than $\Delta$ from the ET, fully reflect the received signal ($\beta_i=1$), otherwise remain inactive ($\beta_i=0$). In this case, the average harvested energy can be written as
\begin{align}
\mathcal{Q}_{\text{DBB}}=\epsilon \Lambda(\xi,\Delta)+(1-\epsilon)\big[\Lambda(\Delta,\rho)+Q_{\text{FB}}(\Delta,\lambda)\big],
\end{align}
where the first and second terms refer to the location range $(\xi, \Delta)$ and $(\Delta, \rho)$, respectively, and $\epsilon=(\xi^2-\Delta^2)/(\rho^2-\xi^2)$. The design parameter $\Delta$ can be used as a means to provide an energy harvesting balance among the ERs. More specifically, the ERs that are close to the ET ($d_i\leq \Delta$) harvest sufficient energy from the omnidirectional component and facilitate the energy harvesting of the distant ERs by keeping silent. We can design $\Delta$ based on the max-min fairness that maximizes the minimum average harvested energy and is given as follows
\begin{align}
&\Delta^*\!=\!\arg_{\Delta \in (\xi,\rho)}\!\max \min\! \left\{\!Q_{\text{OM}}(\Delta), Q_{\text{OM}}(\rho)+\bar{Q}_{\text{RE}}^{\Phi(\Delta)}(\pmb{1},\rho) \!\right\} \nonumber \\
&\!=\!\left\{\Delta \!\left|\!\begin{array}{ll} Q_{\text{OM}}(\xi)=Q_{\text{OM}}(\rho)+\bar{Q}_{\text{RE}}^{\Phi(\Delta)}(\pmb{1},\rho)\;\;\text{If}\;\mathcal{T}\;\text{is true},\\
Q_{\text{OM}}(\Delta)=Q_{\text{OM}}(\rho)+\bar{Q}_{\text{RE}}^{\Phi(\Delta)}(\pmb{1},\rho)\;\text{elsewhere},  \end{array} \right. \right. \label{deltaopt}
\end{align}
where $\mathcal{T}=\{Q_{\text{OM}}(\xi)\leq Q_{\text{OM}}(\rho)+\bar{Q}_{\text{RE}}^{\Phi(\rho)}(\pmb{1},\rho) \}$, and $\Phi(\Delta)=\{x_i \in \Phi: ||x_i||\geq \Delta \}$. The above expression takes into account the active/inactive ERs with the worst energy harvesting performance and maximizes the minimum between them; since we have closed form expressions for the harvesting performance for both cases (\eqref{harvesting}, \eqref{Q_bar}), the optimal value $\Delta^*$ can be found by numerically solving the equation in \eqref{deltaopt}. 

\subsubsection{Probabilistic binary backscattering (PBB)}
Similar to the DBB technique, the ERs either fully reflect the received signal or remain inactive. In this case, the mode selection is performed in a probabilistic way and thus the ERs with a probability $p$ operate in backscattering mode and with a probability $1-p$ remain silent. The average harvested energy can be expressed as 
\begin{align}
\mathcal{Q}_{\text{PBB}}(p)= \Lambda(\xi,\rho)+p Q_{\text{FB}}(\xi,\lambda p).
\end{align}
Equivalently to the DBB scheme, $p$  can be designed to boost the energy harvesting performance of the worst case (i.e., ERs located in the edge) such as
\vspace{-0.2cm}
\begin{align}
p^*=\arg_{p \in (0,1)} \max \big\{\bar{Q}_{\text{RE}}^{\Phi'}(\pmb{1},\rho) \big\},
\end{align}
where $\phi'$ denotes a PPP with density $p\lambda$;  due to the complexity of the expression, $p^*$ does not have closed-form solution and can be evaluated numerically.  

\vspace{-0.5cm}
\subsection{Harvesting-target backscattering (HTB)}

In the HTB scheme, we assume that each ER aims to achieve a minimum harvested energy from the retrodirective component of the received signal. Let $\Gamma_i \geq 0$ denote the energy harvesting target for the $i$-th ER. In this case, the problem resembles to the conventional power control problem for time varying channels with user-specific signal-to-interference-plus-noise ratio  requirements, which admits on-line distributed implementation (Foschini-Miljanic algorithm \cite{FOS}). Specifically, the reflection coefficient $\beta_i$ for the $i$-th ER is given by the following iterative algorithm
\vspace{-0.1cm}
\begin{align}
\beta_i (l+1)=\min\left(1, \frac{\Gamma_i}{Q_{\text{RE}}^{\Phi}(\pmb{\beta}(l), d_i)}\beta_i(l) \right).\label{iter}
\end{align}

The convergence properties of the iterative scheme have been studied by Yates \cite{YAT}; it is worth noting that the channel gains between the ET and the ERs remain constant during the iterations (e.g., stationary users with slowly-varying channels). For the case where the problem is feasible and thus all ERs satisfy the minimum harvested energy constraint with $\beta_i\leq 1$, the backscattering coefficients are given in closed form i.e., $\pmb{\beta}^*=(\pmb{I}_K-\pmb{F})^{-1}\pmb{u}$, where $\pmb{u}=\left[\frac{\Gamma_1 \tilde{n}_1}{d_1^{-3\alpha}|g_1|^2},\ldots, \frac{\Gamma_K \tilde{n}_K}{d_K^{-3\alpha}|g_K|^2} \right]$,  and $\pmb{F}$ is a $K\times K$ matrix with $F_{i,j}=\frac{\Gamma_i |g_j|^2 d_{j}^{-2\alpha}}{|g_i|^2 d_{i}^{-3\alpha}}$ and $i,j=1,\ldots, K$ \cite{FOS}.

\begin{figure*}[t]
  \begin{minipage}{0.33\linewidth}
    \includegraphics[width=\linewidth]{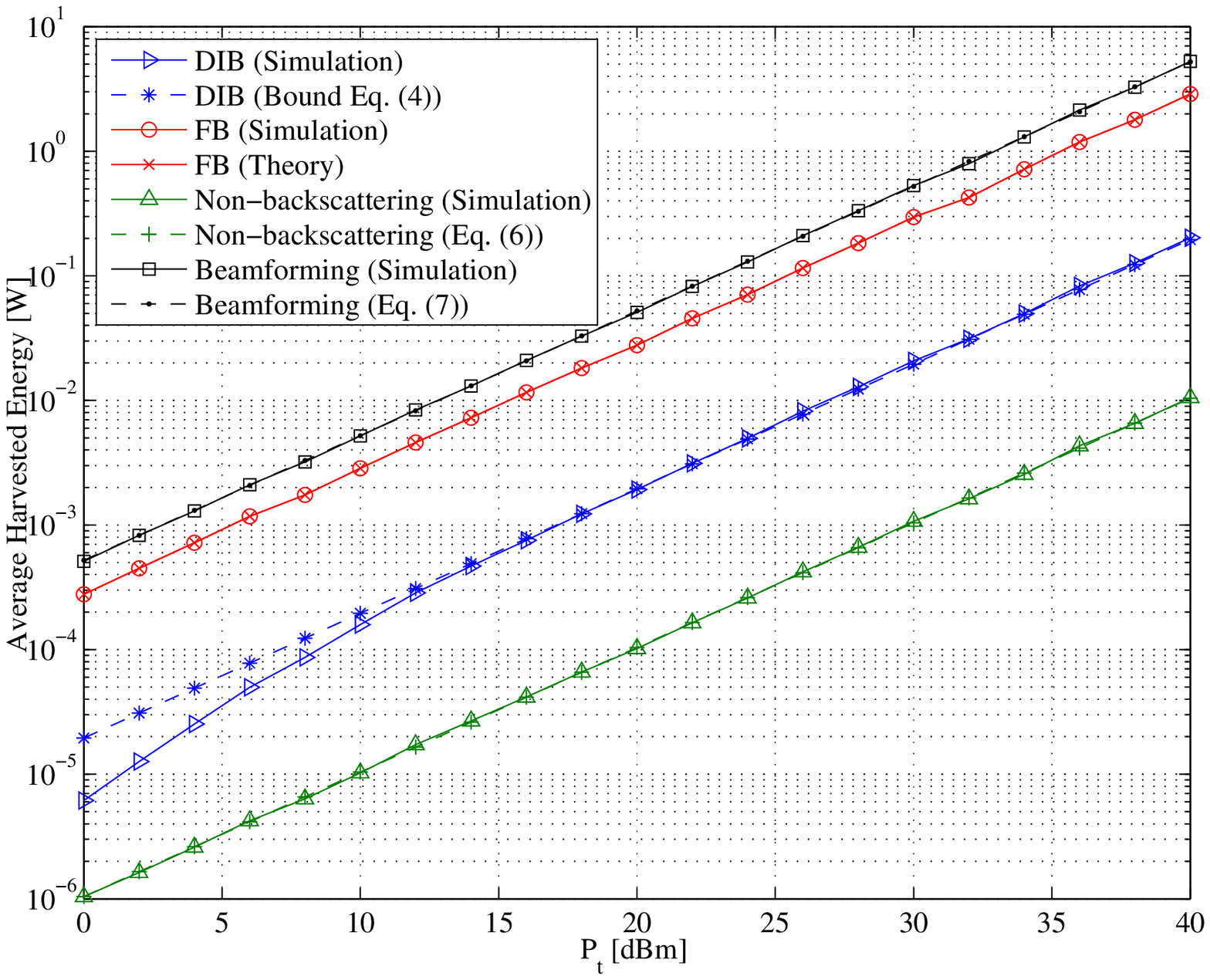}
    \vspace{-0.5cm}
    \caption{Average harvested energy versus the transmit power (DIB, FB, Non-backscattering, Beamforming.).}\label{fig1}
  \end{minipage}
  \begin{minipage}{0.33\linewidth}
    \includegraphics[width=\linewidth]{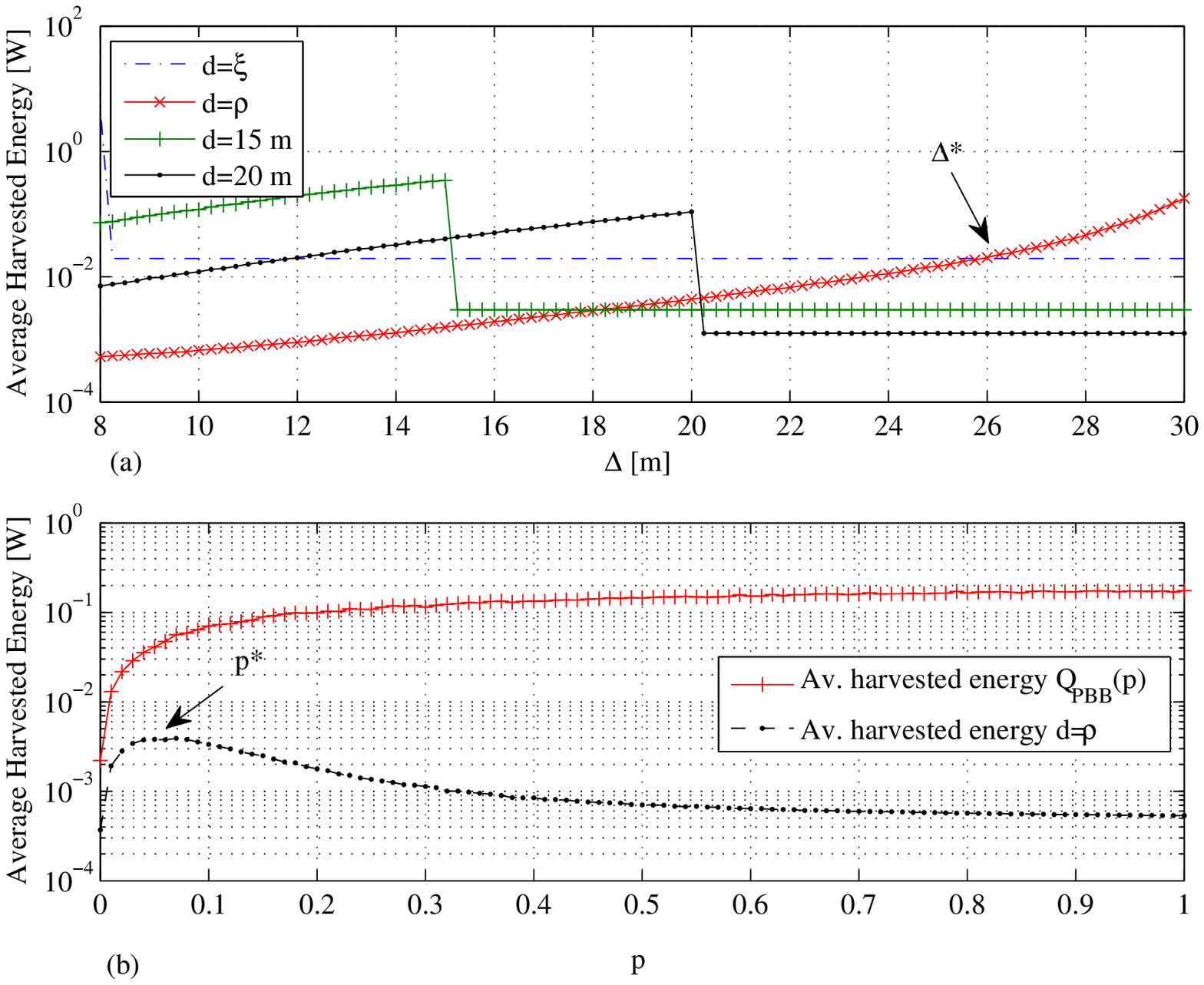}
    \vspace{-0.5cm}
    \caption{(a) DBB scheme- average harvested energy versus $\Delta$, (b) PBB scheme- average harvested energy versus $p$; $\xi=8$ m, $P_t=40$ dBm}\label{fig2}
  \end{minipage}
  \begin{minipage}{0.33\linewidth}
    \includegraphics[width=\linewidth]{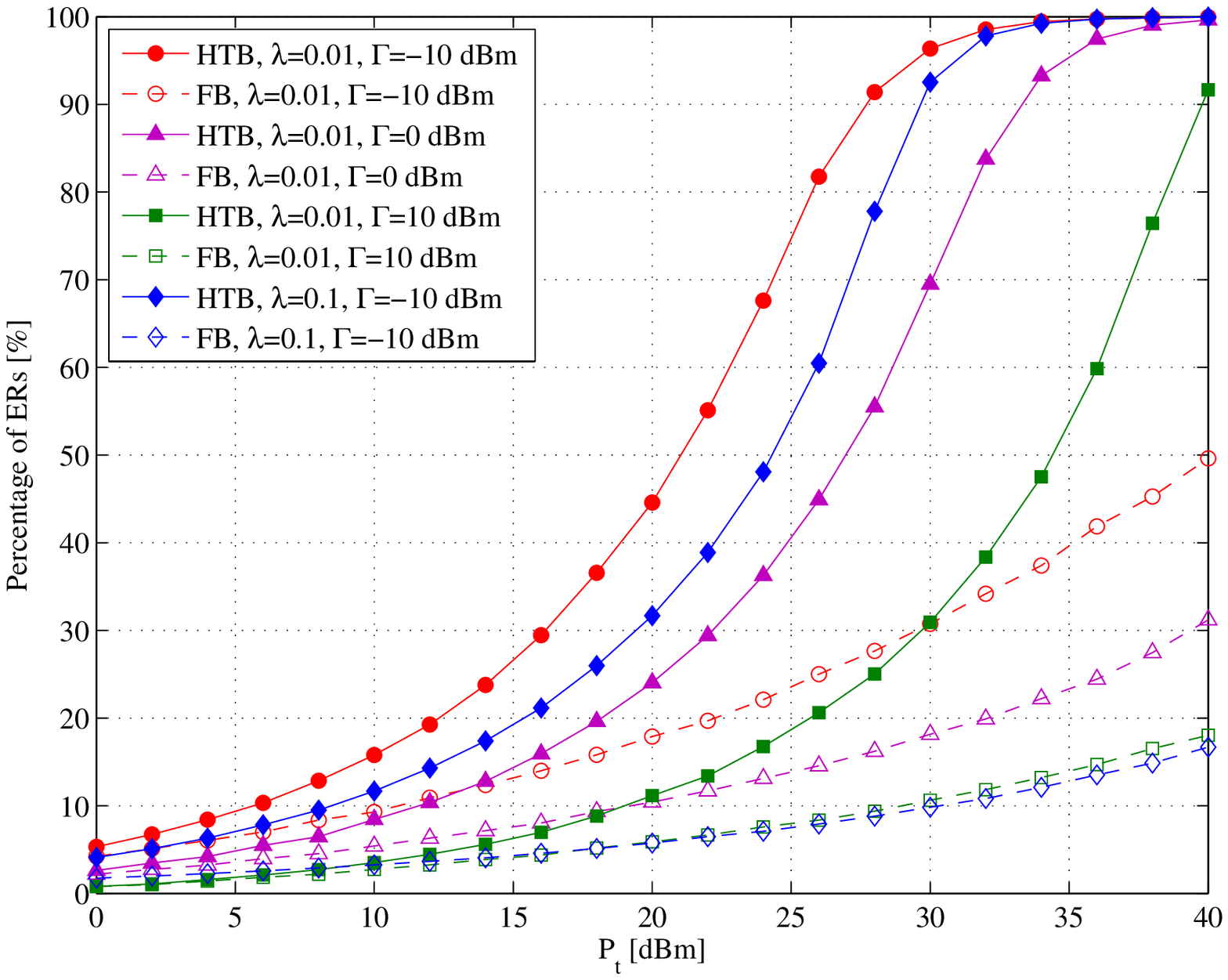}
    \vspace{-0.5cm}
    \caption{Percentage of ERs that satisfy the energy harvesting constraint for the HTB scheme.}\label{fig3}
  \end{minipage}
\end{figure*}

\vspace{-0.3cm}
\section{Numerical results}

The simulation setup follows the description in Sec. \ref{system_model} with $M=500$, $\lambda=0.01$, $f_c=900$ MHz, $\alpha=3$, $\sigma^2=-150$ dBm, $\tau=10^{-8}$ sec, $\xi=2$ m, $\rho=30$ m, and $\zeta=1$ (unless otherwise stated). In Fig. \ref{fig1}, we present the average harvested energy versus the transmit power $P_t$ for the proposed reflection policies. The performance of perfect beamforming (full CSI at the ET) and non-backscattering (omnidirectional transmission) are given as benchmarks. As it can be seen, the FB scheme significantly outperforms the non-backscattering scheme and provides a performance that is close to the perfect beamforming case. On the other hand, the DIB scheme outperforms non-backscattering, while it provides fairness among ERs (in comparison to FB).     

Fig. \ref{fig2}(a) plots the average harvested energy of the DBB scheme for different ER locations versus the distance threshold $\Delta$. As it can be seen, the parameter $\Delta$ provides an energy harvesting balance among ERs located in different distances from the ET. According to the expression in \eqref{deltaopt}, the value $\Delta^*$ that ensures $\max-\min$ fairness corresponds to the cross-over point between the average harvested energy of the ERs located in distance $\xi$ and $\rho$, respectively. In Fig. \ref{fig2}(b), we deal with the performance of the PBB scheme. By controlling the probability $p$, we achieve different energy harvesting balances among the ERs; the optimal value $p^*$ is selected to maximize the energy harvesting at the edge of the cell.   

Finally, Fig. \ref{fig3} deals with the HTB scheme and plots the percentage of the ERs that can
achieve a common harvesting target $\Gamma_i=\Gamma$ from the retrodirective component versus the transmit power $P_t$. The HTB scheme is based on the distributed implementation (see \eqref{iter}) with (maximum) $100$ iterations; the FB scheme is used as a benchmark. We can observe that as the harvesting threshold and the density of users increase, the percentage of the ERs that satisfy the energy harvesting constraint decreases. In addition, the HTB scheme  significantly outperforms the FB scheme and the performance gain becomes more important as the transmit power increases. 

\vspace{-0.4cm}
\section{Conclusion}
This letter has dealt with a new EB scheme for WPT multi-user networks, which jointly combines large antenna retrodirectivity with passive signal backscattering. The proposed technique does not require CSI or active signal transmission at the ERs, and the energy harvesting performance depends on the reflection coefficients. Several reflection strategies have been investigated and analyzed by taking into account location randomness. The proposed EB technique has lower power consumption than conventional counterparts and is a promising solution for future machine-type communications. An interesting extension of this work is to investigate the proposed scheme for more complex networks with multiple ETs.        

\vspace{-0.4cm}
\appendix

The complementary cumulative distribution function (CCDF) of the harvested energy (associated with the retrodirective component of the received signal) for a given path-loss value $d$ is written as
\vspace{-0.15cm}
\begin{align}
&F(x|d)=F(x|d_i=d)=\mathbb{P}(Q_{\text{RE}}^{\Phi}(\pmb{1},d)>x) \nonumber \\
&=\mathbb{P}\left\{|g_i|^2> Y(x)\left( \sum_{x_k \in \Phi, k\neq i} d_k^{-2\alpha}|g_k|^2+ \frac{M\sigma^2}{\zeta P_t\tau}\right)  \right\} \nonumber \\
&=\exp\left( -U(x) \right)\mathbb{E}_{\Phi,g} \exp \left(-\frac{Y(x)}{M}\sum_{x_k \in \Phi, k\neq i} d_k^{-2\alpha}|g_k|^2 \right) \nonumber \\
&=\!\exp(-U(x))\exp\left(\!\!-\lambda\! \int_{\mathcal{D} }\!\! \left[\!1\!-\!\mathbb{E}_{g}\!\exp\!\left(-\frac{Y(x) |g|^2 y^{-2\alpha}}{M} \right)\! \right] \!y dy \!\!\right) \nonumber \\
&=\exp(-U(x)) \exp\left(-2\pi \lambda \int_{\mathcal{\xi}}^{\rho} \frac{Y(x) y}{Y(x)+y^{2\alpha}}  dy \right), \label{pgfl}
\end{align}
where $Y(x)=\frac{\frac{x}{\zeta P_td^{-\alpha}}-1}{(M+1-\frac{x}{\zeta P_td^{-\alpha}})d^{-2\alpha}}$, $U(x)=\frac{Y(x)\sigma^2}{\zeta P_t\tau}$,  $|g_i|^2$ follows an exponential distribution with rate parameter $1/M$ and \eqref{pgfl}  follows from the probability generating functional of a PPP \cite[Sec. 4.6]{HAN2}. From the CCDF, we can calculate the average harvested energy for a given path-loss value $d$ as follows
\vspace{-0.15cm}
\begin{align}
\bar{Q}_{\text{RE}}^\Phi(\pmb{1},d)=\int_{\zeta P_td^{-\alpha}}^{(M+1)\zeta P_td^{-\alpha}}F(x|d)dx. \label{Q_bar}
\end{align}
With expectation over the path-loss values, we have
\vspace{-0.12cm}
\begin{align}
&Q_{\text{FB}}(\xi, \lambda)=\mathbb{E}[\bar{Q}_{\text{RE}}^\Phi(\pmb{1},d)]=\int_{\mathcal{D}} \bar{Q}_{\text{RE}}^\Phi(\pmb{1},y)f_d(y)y dy \nonumber \\
&\;=\frac{2}{\rho^2-\xi^2}\int_{\xi}^{\rho} \int_{\zeta P_t y^{-\alpha}}^{(M+1)\zeta P_t y^{-\alpha}}F(x|y) y dx dy.
\end{align}

\end{document}